\documentclass[a4paper,twocolumn,11pt,accepted=2025-10-16]{quantumarticle}
\pdfoutput=1

\expandafter\let\csname equation*\endcsname\relax
\expandafter\let\csname endequation*\endcsname\relax

\usepackage{enumerate,appendix}
\usepackage{amsmath, amsthm, amssymb,commath}
\usepackage{color,calc,graphicx}
\usepackage[usenames,dvipsnames,svgnames,table,cmyk,hyperref]{xcolor}
\usepackage[colorlinks]{hyperref}
\usepackage{optidef}

\usepackage{graphicx}
\usepackage{amsmath}
\usepackage{latexsym}
\usepackage{bbm}

\usepackage{booktabs}
\usepackage{multirow}
\usepackage{subcaption}
\usepackage{dcolumn}
\usepackage{mathrsfs}

\def \rank{\mathop{\rm rank}}
\def \be {\begin{equation}}
\def \ee {\end{equation}}

\newcommand{\Tr}{\mathrm{Tr}}

\def \sofc2{{\cal S}({\mathbb C}^2)}

\def\>{\rangle}
\def\<{\langle}

\newtheorem{definition}{Definition}
\newtheorem{theorem}{Theorem}
\newtheorem{lemma}[theorem]{Lemma}
\newtheorem{proposition}[theorem]{Proposition}

\def\Label#1{\label{#1}\ [\ \text{#1}\ ]\ }
\def\Label{\label}

\begin{document}

\title{Indefinite causal order strategy does not improve the estimation of group action}

\author{Masahito Hayashi}
\address{School of Data Science, The Chinese University of Hong Kong,
Shenzhen, Longgang District, Shenzhen, 518172, China}
\address{International Quantum Academy, Futian District, Shenzhen 518048, China}
\address{Graduate School of Mathematics, Nagoya University, Furo-cho, Chikusa-ku, Nagoya, 464-8602, Japan}
\orcid{0000-0003-3104-1000}
\email{hmasahito@cuhk.edu.cn, masahito@math.nagoya-u.ac.jp}
\maketitle

\begin{abstract}
We consider estimation of unknown unitary operation when the set of possible unitary operations is given by a projective unitary representation of a compact group.
%We show that indefinite causal order strategy nor adaptive strategy does not improve the performance of this estimation when error function satisfies group covariance.
We show that neither indefinite causal order strategy nor adaptive strategy improves the performance of this estimation when error function satisfies group covariance.
That is, the optimal parallel strategy gives the optimal performance even under indefinite causal order strategy and adaptive strategy.
To study this problem, we newly introduce the concept of 
generalized positive operator valued measure (GPOVM), and its convariance condition.
Using these concepts, we show the above statement.
\end{abstract}


%\maketitle

\section{Introduction}\label{S1}
Recently, indefinite causal order strategies and adaptive strategies have attracted significant attention in quantum information theory \cite{BMQ}. The papers \cite{GE,GDP,HHLW,PLLP,ZP,RM,KW} investigated the advantages of adaptive strategies for quantum channel discrimination without considering the asymptotic regime. Similarly, the paper \cite{Chi} explored these advantages using indefinite causal order strategies. However, in the asymptotic setting, the benefit of adaptive strategies becomes less straightforward. The works \cite{MBHK,WW,FFRS} demonstrated that adaptive strategies provide no improvement for the asymmetric scenario of quantum channel discrimination, specifically the Stein-type bound, which extends the classical result from \cite{Ha09} to the quantum domain. Additionally, the paper \cite{SHW} showed that adaptive strategies do not enhance discrimination performance in the symmetric setting for classical-quantum channels, where the input is classical and the output is quantum, also as an extension of \cite{Ha09}. 
However, the papers \cite{HHLW,SHW,IM,BMQ22} established that adaptive strategies do offer an advantage in the symmetric setting for fully quantum channel discrimination, where both the input and output systems are quantum.
In contrast, the paper \cite{BMQ22} 
showed that indefinite causal order strategies cannot outperform parallel strategies for unitary discrimination under the uniform prior distribution
when the possible unitaries forms a finite group.

A more practical scenario in this context is channel estimation. In studies using parallel strategies, quantum Fisher information and Cramér-Rao-type bounds simplify the problem to optimizing these quantities based on input state choice, a topic covered extensively in \cite{Fujiwara-2003,Fujiwara-2004,H11,Giovannetti11,Giovannetti06,Giovannetti,Jonathan,Imai,Okamoto_2008,Nagata07,Thomas-Peter,Hayashi2024}. Recently, the papers \cite{AY,LHYY,KGAD,ZYC,LHYY24} examined the advantages of adaptive and indefinite causal order strategies in terms of Fisher information. 
In particular, the paper \cite{ZYC} showed that 
the super-Heisenberg scaling can be achieved
under indefinite causal order strategies.
However, it is well-known that Fisher information does not fully describe the asymptotic behavior of estimation errors under Heisenberg scaling. To capture this behavior, the direct evaluation of estimation errors, bypassing Fisher information and Cram\'{e}r-Rao-type bounds, is necessary. When the set of unknown unitaries forms a (projective) unitary representation of a group, group representation theory provides a powerful approach \cite{BDM,Luis,BBM04,CDPS,CDS,Imai_2009,H06,MH16-9,Chiribella-2011}. This method yields optimal performance without requiring asymptotic approximations. 
Although it predicts Heisenberg scaling, the optimal coefficient obtained differs from that predicted by Cram\'{e}r-Rao methods \cite{H06,H11,HVK,HLY,gorecki2020pi}. Recently, \cite{Haya04} highlighted this discrepancy from the perspective of mutual information, underscoring the limited applicability of Cramér-Rao approaches under Heisenberg scaling.
In fact, the paper \cite{GDP} showed no advantages of adaptive strategies for quantum channel estimation
under the covariant setting, but their analysis is limited to the case  when the channel fidelity \cite{RAGINSKY200111} is adopted for the estimation precision. 
General covariant error function was not studied in this scenario.
Further, no existing paper studied this problem for indefinite causal order strategies.

This paper focuses on unitary operation estimation in a general setting, where the model forms a projective unitary representation of a group $G$. We address this problem within a comprehensive framework that includes all indefinite causal order strategies \cite{BMQ}. Surprisingly, even under this broad framework, the conventional parallel strategy remains optimal when the error function satisfies a covariance condition. In other words, the optimal performance achievable by any strategy is also attainable using a parallel strategy.

Our approach is as follows. We define a generalized positive operator-valued measure (GPOVM) for unitary operations associated with a projective unitary representation. Extending the concept of group covariant measurements from \cite{Hol-Cov} and \cite[Chapter 4]{Holevo} to GPOVMs, we develop a GPOVM version of the quantum Hunt-Stein theorem \cite{Hol-Cov}, \cite[Theorem 4.3.1]{Holevo}. This theorem establishes that the optimal performance is achieved by a covariant GPOVM. Finally, we demonstrate that this performance is attainable using a parallel strategy with a detailed symmetric structure as described in \cite{MH16-9}. In our analysis, we show that the action of any covariant GPOVM can be simulated using a parallel strategy. 
This simulation is configured as a covariant measurement in a standard form, utilizing the input state in accordance with the given covariant GPOVM.

The remainder of this paper is structured as follows. Section \ref{S2} formulates our problem, encompassing both indefinite causal order and adaptive strategies \cite{BMQ,AY}. Section \ref{S3} provides a detailed analysis of the problem's structure. Section \ref{S4} demonstrates that the optimal performance under the framework described in Section \ref{S2} can be achieved using a parallel strategy.

\section{Formulation}\label{S2}
Given a compact group $G$ with Haar measure $\mu$, 
we consider $n$ projective representations $f_j$ on the quantum system
${\cal H}_j$ for $j=1, \ldots, n$.
We are allowed to employ various strategies including 
adaptive strategy and indefinite causal order strategy.
To describe our general strategies,
we describe each channel, i.e., each unitary action by 
$d_j$ times of its Choi representation, where
$d_j$ is the dimension of the $j$-th system ${\cal H}_j$.
We describe the unitary $f_j(g)$ 
by $|f_j(g)\rangle\rangle\langle \langle f_j(g)|$.
Here, given a matrix $A=\sum_{k,k'} a_{k,k'}|k\rangle \langle k'|$
on the space ${\cal H}'$ spanned by $\{|k\rangle \}_k$,
the vector $|A\rangle \rangle \in {\cal H}' \otimes {\cal H}'$
is defined as
\begin{align}
|A\rangle\rangle:= \sum_{k,k'} a_{k,k'} |k,k'\rangle.\Label{AHY}
\end{align}

In this case, we denote the input and output systems by 
${\cal H}_{j,I}$ and ${\cal H}_{j,O}$, respectively.
For example, when ${\cal H}={\cal H}_1 \oplus {\cal H}_2$, we have
\begin{align}
|I_{\cal H}\rangle\rangle= 
|I_{{\cal H}_1}\rangle\rangle \oplus |I_{{\cal H}_2}\rangle\rangle. \label{NJ1}
\end{align}
When ${\cal H}={\cal H}_1 \otimes {\cal H}_2$, we have
\begin{align}
|I_{\cal H}\rangle\rangle= 
|I_{{\cal H}_1}\rangle\rangle \otimes |I_{{\cal H}_2}\rangle\rangle. \label{NJ2}
\end{align}

Now, the whole system is 
$(\otimes_{j=1}^n {\cal H}_{j,I})\otimes (\otimes_{j=1}^n {\cal H}_{j,O})$,
the group $G$ acts only on the output system 
${\cal H}:=(\otimes_{j=1}^n {\cal H}_{j,O})$, and the input system
${\cal K}:=(\otimes_{j=1}^n {\cal H}_{j,I})$ 
can be considered as multiplicity space.
We denote the dimension of ${\cal H}$ by $d$. So, $\dim{\cal K}=d$. 
The application of $g$ on ${\cal H}$ is written as
$f(g):= \otimes_{j=1}^n  f_j(g)$ and its state is written as
$|f(g)\rangle\rangle \langle\langle f(g)|$.
The average with respect to the Haar measure $\mu$ is 
\begin{align}
\rho_{\mu}:=&\int_{G}
|f(g)\rangle\rangle \langle\langle f(g)|\mu(dg).
\end{align}

Our strategy is written as a set of measurement operators $\{M_k\}$, 
where a measurement operator $M_k$ is given as a positive semi-definite operator on ${\cal K} \otimes {\cal H}$.
Since the full probability needs to sum to
$1$ for any element $g \in G$, 
the set of measurement operators $\{M_k\}$ needs to satisfy
the condition
\begin{align}
\Tr \Big(\sum_{k}M_k \Big)\rho_{\mu}=1\label{AA}.
\end{align}
In fact, indefinite causal order strategy 
and adaptive strategy satisfy %at least 
the condition \eqref{AA}
even when $n$ is larger than $2$.
This is because the full probability needs to be $1$
when the channel is prepared as the average of $f(g)$.
Therefore, it is natural to impose 
the condition \eqref{AA} to possible strategies.

Here, in order to write down these conditions in a general way,
using a linear function $F$ from the set of Hermitian matrices on 
${\cal K}_A \otimes {\cal K}_B$ to a linear space $V$ and 
an element $C\in V$,
we consider a general linear condition
\begin{align}
F(\sum_x M_x)=C. \label{NM}
\end{align}
For example, the condition for adaptive strategies is given 
in \cite[Algorithm 1]{AY},
\cite[Section III-B]{BMQ} as
\begin{align}
&\Tr(\sum_x M_x)=\dim {\cal K}_B  \label{AD0}
\\
&\Tr_{(A,k+1),\ldots,(A,n),(B,k+1),\ldots,(B,n)} (\sum_x M_x)\notag\\
=&
I_{(B,k)}\Tr_{(A,k+1),\ldots,(A,n),(B,k),(B,k+1),\ldots,(B,n)} (\sum_x M_x)
 \label{AD}
\end{align}
for $k=1, \ldots,n$.
That is, while a POVM $M=\{M_x\}_x$
satisfies the condition $\sum_x M_x=I $,
our measurement satisfies the condition \eqref{NM}.
We call this measurement a generalized positive valued measure (GPOVM)
with the condition \eqref{NM}.

In contrast, 
the paper \cite[Section III-C]{BMQ22} gives the condition for 
an indefinite causal order strategy as
\begin{align}
\Tr (\sum_x M_x) (C_1 \otimes C_2 \otimes \cdots \otimes C_n)
=1
\end{align}
for any $n$ Choi matrices $C_1 , C_2, \cdots, C_n$.
While the above condition is given as a general linear condition \eqref{NM}, it has complicated form.
For example, 
the paper \cite[Eqs. (4),(5),(6)]{BMQ}
writes down the form \eqref{NM} for $n=2$
and the paper \cite[Appendix A]{BMQ22} does it for $n=3$.

In the following, we study the optimization of our measurement to estimate the unknown 
action $g\in G$
under the condition \eqref{AA}.
For this aim, we introduce covariant error function
$w(g,\hat{g})$, where $g$ is the true action and $\hat{g}$ is our guess.
Since the above case addresses only the case with discrete outcomes,
we extend the condition \eqref{AA} to the case with continuous outcomes.
Our measurement is given as a generalized 
positive operator valued measure (GPOVM)
$M$ over the group $G$, which is formulated as a generalization of 
positive operator valued measure (POVM) \cite{Holevo} as follows.

\begin{definition}[Generalized positive operator valued measure]
We denote the set of Borel sets of $G$ by ${\cal B}(G)$.
A generalized 
positive operator valued measure (GPOVM) $M$ is given as a map from ${\cal B}(G)$
to the set of positive semi-definite operators on ${\cal K}\otimes {\cal H}$.
It needs to satisfy the conditions
\begin{align}
M(\emptyset)=0, \quad \Tr M(G)\rho_\mu=1.\label{NB1}
\end{align}
and
\begin{align}
M(\cup_{j}B_j)=\sum_{j}M(B_j)\label{NB2}
\end{align}
for countable Borel sets $\{B_j\}_j$
with the condition $B_j\cap B_{j'} = \emptyset $ with $j\neq j'$.
\end{definition}

We introduce the covariant measurement, whose POVM version was introduced in
\cite{Hol-Cov}, \cite[Chapter 4]{Holevo}.

\begin{definition}[Covariant]
We say that a GPOVM $M$ is covariant when
\begin{align}
M(B_g):= f(g)^\dagger M(B) f(g)
\end{align}
for $B \in {\cal B}(G)$ and 
\begin{align}
B_g:= \{ g g'\}_{g'\in B}.\label{NMA}
\end{align}
\end{definition}

We consider 
the subspace ${\cal H}_0$ spanned by
$\{|f(g)\rangle\rangle\}_{g \in G}$
and denote the projection to this subspace ${\cal H}_0$ by $ P_0$.
% \langle\langle f(g)|\mu(dg).
As shown in Appendix \ref{A1}, we have the following lemma. 
\begin{lemma}\Label{LE1}
There exists a positive semi-definite matrix $T$ 
on ${\cal H}_0$ 
such that
\begin{align}
P_0 M(B) P_0 &= \int_B f(g)^\dagger T f(g) \mu(d g),\Label{SLP} \\
\Tr T \rho_\mu &= 1\label{NB8}.
\end{align}
\end{lemma}
Due to the relation
$\Tr M(B) |f(g)\rangle\rangle \langle\langle f(g)|
=\Tr P_0 M(B) P_0 |f(g)\rangle\rangle \langle\langle f(g)|$,
it is sufficient to treat $P_0 M(B) P_0$ instead of $M(B)$.

The condition \eqref{NB8}
will be simplified in the end of Section \ref{S3}.
When $M_e$ is given as a positive semi-definite operator $T$,
the above GPOVM is written as $M_T$.

To study the precision of our measurement,
we employ 
the estimation error $w(g,\hat{g})$ 
between the true value $g$ and our estimate $\hat{g}$.
When our measurement is given by GPOVM $M$,
the average of $w$ with the true $g$ is given as
\begin{align}
D_{w,g}(M):=
\int_G w(g,\hat{g}) \Tr M(d \hat{g}) 
|f(g)\rangle\rangle\langle\langle f(g)|\label{SB1}
\end{align}
because $\Tr M(d \hat{g}) |f(g)\rangle\rangle\langle\langle f(g)|$
describes the distribution of the estimate $\hat{g}$
when the true value is $g$.
The Bayesian average under the Haar measure $\mu$ is given as 
\begin{align}
D_{w,\mu}(M):=
\int_G D_{w,g}(M)\mu(d g).\label{SB2}
\end{align}
The minimax criterion focuses on the following value
\begin{align}
D_{w}(M):=\max_{g\in G} D_{w,g}(M).\label{SB3}
\end{align}

Holevo \cite[Theorem 2]{Hol-Cov}, \cite[Theorem 4.3.1]{Holevo}
showed the equivalence between the 
minimax criterion and the Bayeian criterion with Haar measure in the state estimation, which is called the quantum Hunt-Stein theorem (Also, see 
\cite[Theorem 4.1]{Group2}).
As shown below, we establish its extension to the GPOVM version, which 
shows the equivalence between 
the minimization of $D_{w,\mu}(M)$ and the minimization of $D_{w}(M)$.

\begin{theorem}\label{Th1}
When the estimation error $w(g,\hat{g})$ satisfies
the condition 
\begin{align}
w(g,\hat{g})=w(g'g,g'\hat{g})%=w(g g',\hat{g} g')
\label{NMG}
\end{align}
for $g,\hat{g},g'\in G$,
we have
\begin{align}
&\min_{M:\hbox{GPOVM}}D_{w,\mu}(M)\notag\\
=&\min_{M:\hbox{cov. GPOVM}}D_{w,\mu}(M) \notag\\
=&\min_{T: \rank T=1, \eqref{NB8}}
D_{w,\mu}(M_T) \notag\\
=&\min_{M:\hbox{GPOVM}}D_{w}(M)\notag\\
=&\min_{M:\hbox{cov. GPOVM}}D_{w}(M).
\end{align}
\end{theorem}

\begin{proof}
Since the set $\{M_T:\rank T=1, \eqref{NB8}\}$
forms the set of extremal points of the set of 
covariant GPOVMs, we have the relation
\begin{align}
\min_{M:\hbox{cov. GPOVM}}D_{w,\mu}(M) 
=\min_{T: \rank T=1, \eqref{NB8}}
D_{w,\mu}(M_T) .\label{VB1}
\end{align}
Next, we have the relations
\begin{align}
& \min_{M:\hbox{cov. GPOVM}}D_{w}(M) \notag\\
\ge &\min_{M:\hbox{GPOVM}}D_{w}(M) 
\ge \min_{M:\hbox{GPOVM}}D_{w,\mu}(M).\label{VB2}
\end{align}
and
\begin{align}
\min_{M:\hbox{cov. GPOVM}}D_{w,\mu}(M) 
\ge \min_{M:\hbox{GPOVM}}D_{w,\mu}(M).\label{VB3}
\end{align}

For a GPOVM $M$ and an element $g\in G$, 
we define the GPOVM $M_g$ as
\begin{align}
M_g(B):=f(g) M(B_g) f(g)^\dagger \label{ZSI}
\end{align}
for $B \in {\cal B}(G)$.
Then, we have
\begin{align}
& D_{w,g'}(M_g) \notag\\
\stackrel{(a)}{=}&
\int_G w(g',\hat{g}) \Tr M_g(d \hat{g}) 
|f(g')\rangle\rangle\langle\langle f(g')|
 \notag\\
\stackrel{(b)}{=}&
\int_G w(g', g\hat{g}) \Tr f(g) M(d \hat{g}) f(g)^\dagger
|f(g')\rangle\rangle\langle\langle f(g')| \notag\\
\stackrel{(c)}{=}&
\int_G w(g^{-1} g', \hat{g}) \Tr M(d \hat{g}) 
|f(g^{-1}g')\rangle\rangle\langle\langle f(g^{-1} g')|
 \notag\\
=& D_{w,g^{-1} g'}(M), \label{BNZ}
\end{align}
where 
$(a)$ follows from the definition \eqref{SB1} of $D_{w,g'}(M)$
$(b)$ follows from the definitions
\eqref{NMA} and \eqref{ZSI} of $B_g$ and $M_g$,
$(c)$ follows from the condition \eqref{NMG}.
Thus,
\begin{align}
& D_{w,\mu}(M_g) 
\stackrel{(a)}{=}\int_G 
D_{w,g'}(M_g) \mu(d g') \notag\\
\stackrel{(b)}{=}&\int_G 
D_{w,g^{-1} g'}(M) \mu(d g') 
\stackrel{(c)}{=} D_{w,\mu}(M),\label{BNZ2}
\end{align}
where 
$(a)$ follows from the definition \eqref{SB2} of $D_{w,\mu}(M)$,
$(b)$ follows from \eqref{BNZ}, and
$(c)$ follows from the invariance of Haar measure $\mu$.
The covariant GPOVM  $\overline{M}:= \int_G  M_g \mu(dg)$
satisfies 
\begin{align}
D_{w,\mu}( \overline{M})
=\int_G D_{w,\mu}(M_g)\mu(dg)
= D_{w,\mu}(M),
\end{align}
which implies the relation
\begin{align}
&\min_{M:\hbox{GPOVM}}D_{w,\mu}(M)\notag\\
\ge & \min_{M:\hbox{cov. GPOVM}}D_{w,\mu}(M) .
\label{VB4}
\end{align}

Also, we have
\begin{align}
&D_{w}(M_g)
\stackrel{(a)}{=} \max_{g'\in G} D_{w,g'}(M_g)\notag\\
\stackrel{(b)}{=}& \max_{g'\in G} D_{w,g^{-1}g'}(M)
=D_{w}(M),
\end{align}
where
$(a)$ follows from the definition \eqref{SB3} of $D_{w}(M)$ and
$(b)$ follows from \eqref{BNZ}.
The covariant GPOVM  $\overline{M}$
satisfies 
\begin{align}
&D_{w}( \overline{M})
=\max_{g\in G} D_{w,g}(\overline{M})\notag\\
=&\max_{g\in G} \int_G
D_{w,g}(M_{g'}) \mu(dg') \notag\\
\le & 
 \int_G \max_{g\in G}
D_{w,g}(M_{g'}) \mu(dg') \notag\\
= & 
 \int_G 
 D_{w}(M_{g'}) \mu(dg') 
= D_{w}(M),
\end{align}
which implies the relation
\begin{align}
\min_{M:\hbox{GPOVM}}D_{w}(M)
\ge \min_{M:\hbox{cov. GPOVM}}D_{w}(M).\label{VB5}
\end{align}
The combination of 
\eqref{VB1},
\eqref{VB2},
\eqref{VB3},
\eqref{VB4}, and 
\eqref{VB5} implies the desired relations.
\end{proof}

\section{Structure of GPOVM}\label{S3}
To study the minimization in Theorem \ref{Th1},
we need to discuss the structure of GPOVM.
For this aim, we consider structure of 
the representation space.
We denote the representation on the first system 
${\cal H}$ as
\begin{align}
{\cal H}= \bigoplus_{\lambda \in \hat{G}_f}
{\cal U}_\lambda \otimes \mathbb{C}^{n_\lambda},
\end{align}
where 
$n_\lambda$ expresses the multiplicity of the representation 
space ${\cal U}_\lambda$.
Here, we denote the representation on 
${\cal U}_\lambda$ by $f_\lambda$.
We denote the set of irreducible representations appearing in $f$
by $\hat{G}_f$.
Then, the whole system is written as
\begin{align}
&{\cal K}\otimes {\cal H}\notag\\
= &
\Big(\bigoplus_{\lambda \in \hat{G}_f}
{\cal U}_\lambda \otimes \mathbb{C}^{n_\lambda} \Big)
\otimes 
\Big(\bigoplus_{\lambda' \in \hat{G}_f}
\mathbb{C}^{d_{\lambda'}} \otimes \mathbb{C}^{n_{\lambda'}} \Big).
\end{align}
By considering \eqref{NJ1} and \eqref{NJ2},
the application of $g$ 
is written as
\begin{align}
&|f(g)\rangle\rangle=
f(g)|I_{\cal H}\rangle\rangle
=f (g)
 \bigoplus_{\lambda \in \hat{G}_f}
| I_{\lambda}\rangle\rangle \otimes | I_{n_\lambda} \rangle\rangle \notag\\
=&
 \bigoplus_{\lambda \in \hat{G}_f}
f_\lambda (g)| I_{\lambda}\rangle\rangle \otimes | I_{n_\lambda} \rangle\rangle \notag\\
=&
 \bigoplus_{\lambda \in \hat{G}_f}
| f_\lambda (g) \rangle\rangle \otimes | I_{n_\lambda} \rangle\rangle.
\end{align}

The average state $\rho_{\mu}$ with respect to the Haar measure $\mu$ is 
rewritten as
\begin{align}
\rho_{\mu}
=&\bigoplus_{\lambda \in \hat{G}_f}
d_\lambda^{-1} I_\lambda \otimes I_\lambda \otimes 
| I_{n_\lambda} \rangle\rangle\langle\langle I_{n_\lambda}|.
\end{align}
The condition \eqref{NB1} is rewritten as
\begin{align}
\sum_{\lambda \in \hat{G}_f}
d_\lambda^{-1} 
\Tr M(G)
I_\lambda \otimes I_\lambda \otimes 
| I_{n_\lambda} \rangle\rangle\langle\langle I_{n_\lambda}|
=1\label{AAJ}.
\end{align}
This condition can be converted to the following condition for $T$ via 
the condition \eqref{NB8};
\begin{align}
\sum_{\lambda \in \hat{G}_f}
d_\lambda^{-1} 
\Tr T
I_\lambda \otimes I_\lambda \otimes 
| I_{n_\lambda} \rangle\rangle\langle\langle I_{n_\lambda}|
=1\label{AA2}.
\end{align}
This description is useful for considering the relation with parallel scheme.

\section{Parallel scheme}\label{S4}
Next, we discuss the parallel scheme \cite{BDM,Luis,BBM04,CDPS,CDS,Imai_2009,H06,MH16-9,Chiribella-2011}.
The aim of this section is to show that 
the performance 
$\min_{T: \rank T=1, \eqref{NB8}}
D_{w,\mu}(M_T)$
can be realized by Parallel scheme.

To have a large multiplicity for each irreducible 
representation spaces ${\cal U}_\lambda$,
we employ a reference system $\mathbb{C}^{l}$ 
with dimension $l$. Then,
we have
\begin{align}
{\cal H}\otimes \mathbb{C}^{l}
= \bigoplus_{\lambda \in \hat{G}_f}
{\cal U}_\lambda \otimes \mathbb{C}^{l n_\lambda}.\label{NMU}
\end{align}
However, when the input state is a pure state, 
the orbit is restricted into the following space by choosing 
a suitable subspace $\mathbb{C}^{\min(d_\lambda,l n_\lambda)}$
of $\mathbb{C}^{l n_\lambda}$.
That is, our representation space can be considered as follows.
\begin{align}
\bigoplus_{\lambda \in \hat{G}_f}
{\cal U}_\lambda \otimes \mathbb{C}^{\min(d_\lambda,l n_\lambda)}.
\end{align}
In the following, we consider the above case.
We denote the projection to $
{\cal U}_\lambda \otimes \mathbb{C}^{\min(d_\lambda,l n_\lambda)}$
by $P_\lambda$.
When $l \ge d_\lambda/n_\lambda$ for any $\lambda \in \hat{G}_f$,
our representation is given as
\begin{align}
{\cal H}':=\bigoplus_{\lambda \in \hat{G}_f}
{\cal U}_\lambda \otimes \mathbb{C}^{d_\lambda}.\label{FG1}
\end{align}

A parallel strategy is given as a pair of 
a choice of an input state $\rho$ on ${\cal H}'$
and a choice of a POVM $\Pi$ on ${\cal H}'$.
The error probability is written as
\begin{align}
D_w(\rho,\Pi):=
\int_{G} w(g,\hat{g}) \Tr \Pi(d\hat{g}) f(g) \rho f(g)^\dagger.
\end{align}

In particular, as stated below, the reference
\cite[Theorem 1]{MH16-9} showed that 
the optimization of the average of the error $w$ satisfying the condition \eqref{NMG} under a parallel strategy 
can be achieved by the following simple strategy.
The input state is a pure state 
$|\psi \rangle $ on ${\cal H}'$.
The choice of our POVM can be restricted to the following covariant POVM $\Pi_{cov}$;
\begin{align}
\Pi_{cov}(B):= \int_B f(g)^\dagger |F\rangle \langle F| f(g) \mu(dg)
\end{align}
for $B \in {\cal B}(G)$, where
\begin{align}
|F\rangle:= \bigoplus_{\lambda \in \hat{G}_f} \sqrt{d_\lambda} | I_{\lambda}\rangle\rangle.
\end{align}
Hence, this parallel strategy is characterized by
the input state $|\psi \rangle$ and is denoted by 
$S[|\psi \rangle] $.
The error probability is written as
\begin{align}
D_w(|\psi \rangle):=D_w(|\psi \rangle \langle \psi|,\Pi_{cov}).
%:=\int_{G} w(g,\hat{g}) \Tr \Pi_{cov}(d\hat{g}) f(g) |\psi \rangle \langle \psi|f(g)^\dagger.
\end{align}
Then, the preceding result \cite[Theorem 1]{MH16-9}
is stated as follows. 
\begin{proposition}[\protect{\cite[Theorem 1]{MH16-9}}]\label{Pro}
Assume the condition \eqref{NMG}.
Then, we have
\begin{align}
\min_{\rho,\Pi}D_{w,\mu}(\rho,\Pi)
=\min_{|\psi \rangle}D_w(|\psi \rangle).
\end{align}
\end{proposition}

Our result is to state that 
the this value
$\min_{|\psi \rangle}D_w(|\psi \rangle)$ 
gives the optimal value even under much larger choices.
Any vector $|X\rangle$ in ${\cal H}\otimes {\cal K} $
has the following form
\begin{align}
|X\rangle= 
\bigoplus_{\lambda \in \hat{G}_f}
\bigoplus_{\lambda' \in \hat{G}_f}
|X_{\lambda,\lambda'}\rangle\rangle \otimes 
|Y_{\lambda,\lambda'}\rangle
\end{align}
on the system 
$\bigoplus_{\lambda \in \hat{G}_f}
\bigoplus_{\lambda' \in \hat{G}_f}
{\cal U}_\lambda \otimes {\cal U}_{\lambda'} \otimes 
\mathbb{C}^{n_\lambda} 
\otimes 
\mathbb{C}^{n_{\lambda'}} $,
where 
$|Y_{\lambda,\lambda'}\rangle$ is an element of $\mathbb{C}^{n_\lambda} 
\otimes 
\mathbb{C}^{n_{\lambda'}} $
and $X_{\lambda,\lambda'}$ is a matrix on 
${\cal U}_\lambda$ so that
$|X_{\lambda,\lambda'}\rangle\rangle $ is a vector in 
${\cal U}_\lambda \otimes {\cal U}_{\lambda'}$
as the way of \eqref{AHY}.
When the rank-one operator 
$T=|X\rangle \langle X|$ on ${\cal K}\otimes {\cal H} $
satisfies \eqref{NB8},
we choose the pure state $|\psi[T]\rangle$ as
\begin{align}
|\psi[T]\rangle:=
\bigoplus_{\lambda \in \hat{G}_f}
d_\lambda^{-1/2}
 \langle Y_{\lambda,\lambda}|I_{n_\lambda}\rangle \rangle
|X_{\lambda,\lambda}^T\rangle\rangle .
\end{align}
Then, our result is stated as follows.

\begin{theorem}\label{Th2}
The condition \eqref{AA2} guarantees the normalizing condition of
$|\psi[T]\rangle$.
Then, the behavior of the strategy 
$M_T $ is simulated by the above parallel strategy with the input state
$| \psi[T]\rangle$.
That is, we have
\begin{align}
\min_{T: \rank T=1, \eqref{NB8}}
D_{w,\mu}(M_T)
=\min_{|\psi \rangle}D_w(|\psi \rangle).
\end{align}
\end{theorem}

Since Theorem \ref{Th1} addresses a much larger class than 
the parallel scheme,
our result, Theorem \ref{Th2}, recovers Proposition \ref{Pro}
as a much stronger result.

\begin{proof}
We have
\begin{align}
&  \langle X| f(g) |f(e) \rangle \rangle \notag\\
=&
\sum_{\lambda \in \hat{G}_f}
\langle \langle X_{\lambda,\lambda}| \otimes \langle Y_{\lambda,\lambda}|
f(g)| I_{\lambda}\rangle \rangle \otimes |I_{n_\lambda}\rangle \rangle \notag\\
=&
\sum_{\lambda \in \hat{G}_f}
 \langle Y_{\lambda,\lambda}|I_{n_\lambda}\rangle \rangle
 \langle \langle X_{\lambda,\lambda}|
f(g)| I_{\lambda}\rangle \rangle\notag\\
=&
\sum_{\lambda \in \hat{G}_f}
 \langle Y_{\lambda,\lambda}|I_{n_\lambda}\rangle \rangle
\Tr X_{\lambda,\lambda}^T f(g) \notag\\
=&
\sum_{\lambda \in \hat{G}_f}
 \langle Y_{\lambda,\lambda}|I_{n_\lambda}\rangle \rangle
\langle\langle I_{\lambda}| 
f(g) |X_{\lambda,\lambda}^T\rangle \rangle \notag\\
=&
\sum_{\lambda \in \hat{G}_f}
d_\lambda^{1/2}\langle\langle I_{\lambda}| 
f(g)
d_\lambda^{-1/2} \langle Y_{\lambda,\lambda}|I_{n_\lambda}\rangle \rangle
 |X_{\lambda,\lambda}^T\rangle \rangle \notag\\
=&
\langle F| f(g) | \psi[T]\rangle .
\end{align}

Then, we have
\begin{align}
&\Tr f(g)^\dagger T f(g) |f(e)\rangle \rangle \langle \langle f(e)|\notag\\
=&
\langle \langle f(e)| f(g)^\dagger |X\rangle \langle X| f(g) |f(e)\rangle \rangle \notag\\
=&
\langle F| f(g) | \psi[T]\rangle 
\langle \psi[T]| f(g)^\dagger | F\rangle \notag \\
=&
\Tr f(g)^\dagger | F\rangle \langle F| f(g) 
| \psi[T]\rangle \langle \psi[T]| .
\end{align}
Replacing $g$ by $\hat{g}g'$, we have
\begin{align}
&\Tr f(\hat{g})^\dagger T f(\hat{g}) 
|f(g')\rangle \rangle \langle \langle f(g')|\notag\\
=&
\Tr f(\hat{g})^\dagger | F\rangle \langle F| f(\hat{g}) 
f(g')| \psi[T]\rangle \langle \psi[T]|f(g')^\dagger .
\end{align}
Therefore, the behavior of the strategy $M_T $ with the unknown action $g'$
is simulated by the above parallel strategy with the input state
$| \psi[T]\rangle$ with the unknown action $g'$.
That is,
the performance 
$\min_{T: \rank T=1, \eqref{NB8}}
D_{w,\mu}(M_T)$
can be realized by Parallel scheme.
\end{proof}

In the conversion in the above proof,
the freedom of 
the choice of our covariant GPOVM  
is converted to the choice of our input state in the parallel strategy.
Here, the state in the general setting
corresponds to the measurement in the parallel strategy.

\section{Conclusion}\label{S5}
This paper shows that
indefinite causal order strategy nor adaptive strategy 
does not improve the performance of estimation of unknown operation 
given by a projective unitary representation of a compact group $G$.
In fact, the papers \cite{MH16-9} already obtained the optimal estimation
in various examples under the parallel strategy.
These preceding results give the optimal performance even under 
a more general setting including 
both indefinite causal order strategy and
adaptive strategy. 

Finite group is a special case of a compact group,
the paper \cite{HHHH} initiated the discrimination of the unitary operations given by 
a projective unitary representation of a finite group,
and the paper \cite[Theorem 1]{BMQ22} showed that 
neither indefinite causal order strategies nor adaptive strategies improve
this type of unitary discrimination under the uniform prior distribution.
Since our error function covers the case with the delta function 
in the finite group case,
our result recovers the above result by \cite{BMQ22} as a special case.
Although the paper \cite[Theorem 1]{BMQ22} studied only the average 
discrimination error probability, our analysis covers the worst discrimination error probability as well because our analysis contains the minimix criterion.

We have formulated a covariant GPOVM as Definition 2, and 
have shown that it can be written by using an operator $T$ satisfying the condition \eqref{NB8}. The condition \eqref{NB8} can be simplified to \eqref{AA2}.
We can expect that this formulation can be used for more general cases.
However, this paper does not consider the case with noisy channels.
It is an interesting open problem to extend the obtained result to the noisy case. 
Although the generalized Hunt-Stein theorem can be easily extended to the noisy case, the derivation in Section \ref{S4} highly depends on the noiseless setting. 
Therefore, we can expect that 
the noisy case has an advantage of 
indefinite causal order strategy and/or adaptive strategy
over parallel strategy even under the group covariant setting.
It is an interesting open problem whether 
a similar approach can be used in more general situations, such as the noisy case.

The obtained unusefulness of 
indefinite causal order strategy and adaptive strategy
is related to group symmetry.
In fact, similar facts have been proved in the contest of secure network coding.
The preceding studies show that 
adaptive modifications of the input information in each attacked  
edge by the adversary does not improve the information gain by the adversary when all coding operations are given as linear operations
for classical secure network coding \cite{HC1,HOKC,CH,H24}
and quantum secure network coding \cite{KOH,OKH,HS}.
But,
the paper \cite{HC2} shows that this kind of improvement exists
when nonlinear network code is applied.
Therefore, it is another interesting open problem to clarify
why secure network coding and the estimation of group action
commonly have the unusefulness of adaptive strategy.

\section*{Acknowledgement}
The author is thankful to Professor Xin Wang, Dr. Yu-Ao Chen, and Mr. Chenghong Zhu
for helpful discussions and informing the reference \cite{BMQ}.
The author is supported in part by the National Natural Science Foundation of China (Grant No. 62171212) 
and
the General R \& D Projects of 1+1+1 
CUHK-CUHK(SZ)-GDST Joint Collaboration Fund (Grant No.
GRDP2025-022).

\appendix

\section{Proof of Lemma \ref{LE1}}\label{A1}
{\bf Step 1:} Construction of $T$.

We define the $\delta$ neighborhood $U_{\delta,g'} $ of the element $g'$ as
\begin{align}
U_{\delta,g'}:=
\{g \in G| \|f(g)-f(g')\| < \delta\}.
\end{align}
We define the operator $T_\delta:=M(U_\delta)/\mu(U_\delta)$
with $ U_\delta:=U_{\delta,e}$.

Since
\begin{align}
\Tr M( B_g )\rho_\mu
=&
\Tr f(g)^\dagger M( B )f(g) \rho_\mu \notag\\
=&\Tr M(B)\rho_\mu,
\end{align}
we have
\begin{align}
&\Tr \sqrt{\rho_\mu} P_0  M(B) P_0 \sqrt{\rho_\mu}
=\Tr M(B)\rho_\mu \notag\\
=&\int_G \Tr  M( B_g ) \rho_\mu\mu(dg) 
= \Tr \int_G M( B_g )\mu(dg) \rho_\mu\notag\\
=& \Tr \mu(B) M( G )  \rho_\mu
=\mu(B) .\Label{MNM2}
\end{align}
Hence, the maximum eigenvalue of 
$\sqrt{\rho_\mu} P_0  M(B) P_0 \sqrt{\rho_\mu}$
is upper bounded by $\mu(B)$, i.e., 
the maximum eigenvalue of 
$P_0  M(B) P_0 $
is upper bounded by $\frac{\mu(B)}{\lambda_{\min} (\rho_\mu)}$, 
where $\lambda_{\min} (\rho_\mu)$ is the minimum eigenvalue
of $ \rho_\mu$. 
That is,
\begin{align}
P_0 M(B) P_0 \le 
\frac{\mu(B)}{\lambda_{\min} (\rho_\mu)} P_0,\label{ZG1}
\end{align}
Also, \eqref{MNM2} implies 
\begin{align}
\Tr T_\delta\rho_\mu=1\Label{NVB}.
\end{align}
Since $T_\delta\ge 0$,
$T_\delta$ belongs to the compact set
$\{ X\ge 0 | \Tr X\rho_\mu=1\}$.
There exists a sequence $\delta_n$ such that 
$\delta_n \to 0$ and the limit 
$T:=\lim_{n\to 0} P_0T_{\delta_n}P_0$ exists.
Since $T_{\delta_n} $ is positive semi-definite, 
$T$ is also positive semi-definite.
In addition, \eqref{NVB} guarantees \eqref{NB8}.
The remaining issue is the proof of \eqref{SLP}.

{\bf Step 2:} Proof of \eqref{SLP}.

To show \eqref{SLP}, we prepare 
the following notations as Fig. \ref{fig}.
\begin{align}
V_{\delta',g} &:= cl (U_{\delta',g})\setminus U_{\delta',g} \notag \\
R_{\delta',\delta,g} &:=\{ g_1\in U_{\delta',g}|
\min_{g_2 \in V_{\delta',g}} 
\|f(g_1)-f(g_2)\| > \delta\} \notag \\
\tilde{R}_{\delta',\delta,g} &:=
U_{\delta',g} \setminus {R}_{\delta',\delta,g} 
\end{align}
for $\delta'>2\delta$, where
$cl (U_{\delta',g})$ is the closure of the set $U_{\delta',g}$.

\begin{figure}[tbhp]
\begin{center}
\includegraphics[scale=0.5]{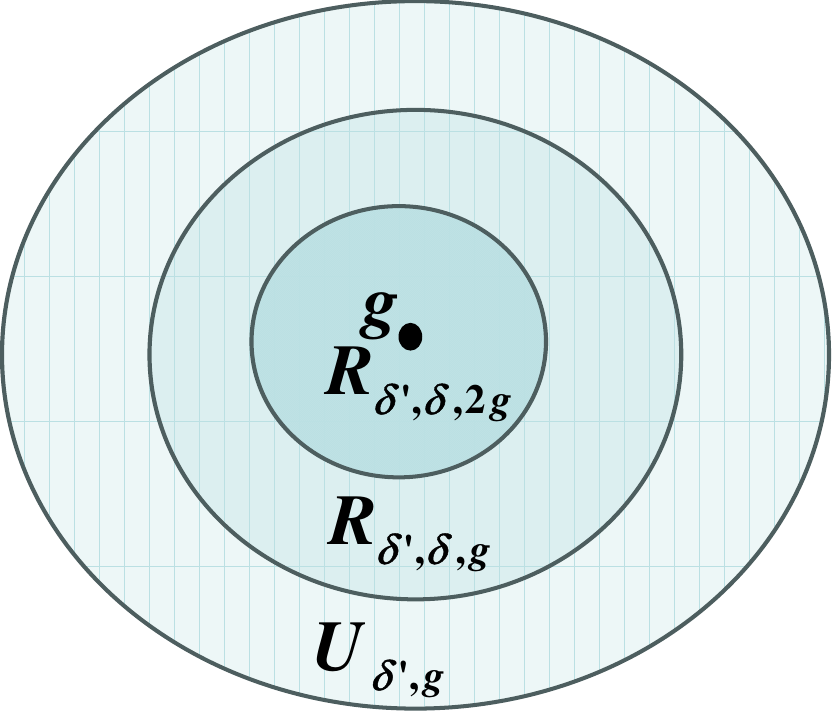}
\end{center}
\caption{Hierarchical structure of sets}
\Label{fig}
\end{figure}

For $g_1 \in R_{\delta',\delta,g}$, we have
$U_{\delta,g_1}\subset U_{\delta',g}$.
Since $g_1 \in U_{\delta,g'}$ for $g'\in U_{\delta,g_1}
\subset U_{\delta',g}$,
we have
\begin{align}
&\mu(U_\delta)\int_{U_{\delta',g}} f(g')^\dagger T_\delta f(g') \mu(dg')\notag\\
=&
\int_{U_{\delta',g}} M(U_{\delta,g'}) \mu(dg') 
\ge 
\mu(U_\delta) M(R_{\delta',\delta,g}),
\end{align}
which implies that
the contribution in the subset $R_{\delta',\delta,g}$
is well reflected in $\int_{U_{\delta',g}} f(g)^\dagger T_\delta f(g)$.

We choose $g'\in \tilde{R}_{\delta',2\delta,g} $,
we have $g'\in U_{\delta',g}$ and 
$\|f(g')-f(g_2)\| > 2 \delta$
for $g_2 \in V_{\delta',g}$.
For $g''\in U_{\delta,g'}$,
we have $\|f(g'')-f(g')\| < \delta$.
Hence,
$\|f(g'')-f(g_2)\| > \|f(g')-f(g_2)\|- \|f(g'')-f(g')\|
> \delta$, which implies 
$U_{\delta,g'} \subset R_{\delta',\delta,g}$.
That is, for $g'\in \tilde{R}_{\delta',2\delta,g} $,
$U_{\delta,g'}\setminus R_{\delta',\delta,g}$ is empty.
Thus,
\begin{align}
0 \le &
\mu(U_\delta) \int_{U_{\delta',g}} f(g')^\dagger T_\delta f(g') \mu(dg')\notag\\
&-\mu(U_\delta)  M(R_{\delta',\delta,g}) 
\notag \\
=&
\int_{U_{\delta',g}} M(U_{\delta,g'}\setminus R_{\delta',\delta,g}) \mu(dg') \notag\\
=&
\int_{R_{\delta',2\delta,g}} \!\! M(U_{\delta,g'}\setminus R_{\delta',\delta,g}) \mu(d g') \notag\\
\le &
\int_{R_{\delta',2\delta,g}} \!\!\!\!M(U_{\delta,g'}) \mu(dg') .
\label{ZG2}
\end{align}
Since $\mu(U_{\delta,g'})=\mu(U_{\delta})$,
the combination of
\eqref{ZG1} and \eqref{ZG2} yields
\begin{align}
0 \le &
\mu(U_\delta) \int_{U_{\delta',g}} 
f(g')^\dagger P_0 T_\delta P_0  f(g') \mu(dg')
\notag \\
&-\mu(U_\delta) P_0  M(R_{\delta',\delta,g}) P_0 
\notag \\
\le &
\int_{R_{\delta',2\delta,g}} P_0 M(U_{\delta,g'}) P_0 \mu(dg') \notag\\
\le &
\int_{R_{\delta',2\delta,g}} 
\frac{\mu(U_{\delta})}{\lambda_{\min} (\rho_\mu)} P_0
 \mu(dg') \notag\\
= & \frac{\mu(U_{\delta})\mu(\tilde{R}_{\delta',2\delta,g})}{\lambda_{\min} (\rho_\mu)} P_0,
\end{align}
which implies 
\begin{align}
0 \le &
P_0 M(U_{\delta',g}) P_0-
P_0 M(R_{\delta',\delta,g})P_0 \notag\\
\le & \frac{\mu(\tilde{R}_{\delta',2\delta,g})}{\lambda_{\min} (\rho_\mu)} P_0.
\end{align}
Thus, we have
\begin{align}
&\Bigg\|
P_0 M(U_{\delta',g}) P_0-
\int_{U_{\delta',g}} f(g)^\dagger P_0 T_\delta P_0 f(g)
\mu(dg)
\Bigg\| \notag \\
\le & \frac{\mu(
\tilde{R}_{\delta',2\delta,g})}{\lambda_{\min} (\rho_\mu)} .
\end{align}
Since $\mu(\tilde{R}_{\delta',2\delta_n,g})\to 0$ 
and $P_0 T_{\delta_n}P_0 \to T$ as $n\to \infty$,
we have
\begin{align}
P_0 M(U_{\delta',g})P_0=
\int_{U_{\delta',g}} f(g)^\dagger T f(g)\mu(dg).
\end{align}
Since any Borel set can be generated by 
sets $\{U_{\delta',g}\}$,
we have \eqref{SLP} for any Borel set $B$.

\bibliographystyle{quantum}
\bibliography{references2}

\end{document}